\documentclass[letterpaper, 10 pt, conference]{ieeeconf} 

\IEEEoverridecommandlockouts  
\usepackage{amsmath}
\usepackage{amssymb}
\usepackage{amsthm}
\usepackage{mathtools}

\usepackage[noadjust]{cite}

\newtheorem{theorem}{Theorem}
\newtheorem{lemma}{Lemma}

\newtheorem{corollary}{Corollary}

\theoremstyle{definition}
\newtheorem{definition}{Definition}
\newtheorem{assumption}{Assumption}
\newtheorem{remark}{Remark}

\newcommand{\R}{\mathbb{R}}

\usepackage{xcolor}
\definecolor{pastelMagenta}{HTML}{FF48CF}
\definecolor{pastelPurple}{HTML}{8770FE}
\definecolor{pastelBlue}{HTML}{1BA1EA}
\definecolor{pastelSeaGreen}{HTML}{14B57F}
\definecolor{pastelGreen}{HTML}{3EAA0D}
\definecolor{pastelOrange}{HTML}{C38D09}
\definecolor{pastelRed}{HTML}{F5615C}

\definecolor{myBlue1}{RGB}{49, 114, 174}
\definecolor{myRed1}{RGB}{224, 107, 97}
\definecolor{myGreen1}{RGB}{68, 156, 118}
\definecolor{myPurple1}{RGB}{117, 112, 173}
\definecolor{myYellow1}{RGB}{221, 162, 66}
\definecolor{myMagenta1}{RGB}{202, 98, 159}
\definecolor{myCyan1}{RGB}{114, 179, 224}

\definecolor{myBlue2}{cmyk}{0.8336, 0.5245, 0.0745, 0}
\definecolor{myRed2}{cmyk}{0.0808, 0.7143, 0.5968, 0.0028}
\definecolor{myGreen2}{cmyk}{0.7402, 0.1786, 0.0667, 0.0021}
\definecolor{myPurple2}{cmyk}{0.6079, 0.5931, 0.0401, 0}
\definecolor{myYellow2}{cmyk}{0.1289, 0.3792, 0.8644, 0.0014}
\definecolor{myMagenta2}{cmyk}{0.1821, 0.7493, 0.0432, 0}
\definecolor{myCyan2}{cmyk}{0.5226, 0.1608, 0.0075, 0}

\definecolor{juliaBlue}{rgb}{0.255, 0.388, 0.847}
\definecolor{juliaRed}{rgb}{0.796, 0.235, 0.2}
\definecolor{juliaGreen}{rgb}{0.22, 0.596, 0.149}
\definecolor{juliaPurple}{rgb}{0.706, 0.322, 0.804}

\usepackage{tikz}
\usepackage{pgfplots}
\pgfplotsset{compat=newest}
\pgfplotsset{every axis legend/.append style={legend cell align=left, font=\small, draw=none, fill=none}}
\pgfplotsset{every axis/.append style={axis background/.style={fill=white}}}
\pgfplotsset{every axis plot post/.append style={
	smooth,
	line width=2pt,
	line join = round,
	mark = none,
	}
}

\pgfplotscreateplotcyclelist{mycolors}{
	{myBlue1, mark=none},
	{myRed1, mark=none},,
	{myGreen1, mark=none}
	{myPurple1, mark=none},
	{myYellow1, mark=none},
	{myMagenta1, mark=none},
	{myCyan1, mark=none}
}

\pgfplotscreateplotcyclelist{mycolors}{%
	solid, myBlue1, mark=none\\%
	solid, myRed1, mark=none\\%
	solid, myGreen1, mark=none\\%
	solid, myPurple1, mark=none\\%
	solid, myYellow1, mark=none\\%
	solid, myMagenta1, mark=none\\%
	solid, myCyan1, mark=none\\%
}

\pgfplotscreateplotcyclelist{mycolors2}{%
	solid, myBlue2, mark=none\\%
	solid, myRed2, mark=none\\%
	solid, myGreen2, mark=none\\%
	solid, myPurple2, mark=none\\%
	solid, myYellow2, mark=none\\%
	solid, myMagenta2, mark=none\\%
	solid, myCyan2, mark=none\\%
}

\pgfplotscreateplotcyclelist{juliacolors}{%
	solid, juliaBlue, mark=none\\%
	solid, juliaRed, mark=none\\%
	solid, juliaGreen, mark=none\\%
	solid, juliaPurple, mark=none\\%
}

\pgfplotsset{
tick label style={font=\footnotesize},
label style={font=\normalsize},
legend style={font=\small},
}

\pgfplotsset{every axis/.append style={
thick,
tick style={semithick}}}

\usepgfplotslibrary{external}
\title{\bf{High Order Robust Adaptive Control Barrier Functions and Exponentially Stabilizing Adaptive Control Lyapunov Functions}}
\author{Max H. Cohen and Calin Belta%
\thanks{The authors are with the Department of Mechanical Engineering, Boston
University, 110 Cummington Mall, Boston, MA 02215, United States \{\texttt{maxcohen,cbelta}\}@\texttt{bu.edu.} This work is supported by the NSF under grants DGE-1840990 and IIS-2024606. Any opinions, findings, conclusions or recommendations expressed in this material are those of the author(s) and do not necessarily reflect the views of the NSF.}
}
\date{}

\begin{document}

\maketitle

\begin{abstract}
This paper studies the problem of utilizing data-driven adaptive control techniques to guarantee stability and safety of uncertain nonlinear systems with high relative degree. We first introduce the notion of a High Order Robust Adaptive Control Barrier Function (HO-RaCBF) as a means to compute control policies guaranteeing satisfaction of high relative degree safety constraints in the face of parametric model uncertainty. The developed approach guarantees safety by initially accounting for all possible parameter realizations but adaptively reduces uncertainty in the parameter estimates leveraging data recorded online. We then introduce the notion of an Exponentially Stabilizing Adaptive Control Lyapunov Function (ES-aCLF) that leverages the same data as the HO-RaCBF controller to guarantee exponential convergence of the system trajectory. The developed HO-RaCBF and ES-aCLF are unified in a quadratic programming framework, whose efficacy is showcased via two numerical examples that, to our knowledge, cannot be addressed by existing adaptive control barrier function techniques.
\end{abstract}

\section{Introduction}
The problem of developing control policies that guarantee stability and safety of nonlinear control systems has received significant attention in recent years. In particular, the unification of Control Lyapunov Functions (CLFs) \cite{SontagSCL89,AmesTAC14} and Control Barrier Functions (CBFs) \cite{AmesTAC17,AmesECC19} has provided a pathway towards safe and stable control of complex nonlinear systems such as autonomous vehicles \cite{AmesCDC14,WeiAutomatica21}, multi-agent systems \cite{AmesTRO17}, and bipedal robots \cite{AmesCDC16}. Although powerful, the guarantees afforded by these approaches are model-based, hence the success in transferring such guarantees to real-world systems is inherently tied to the fidelity of the underlying system model. Inevitably, such models are only an approximation of the true system due to parametric uncertainties and unmodeled dynamics, thus there is strong motivation to study the synthesis of CLF and CBF-based controllers in the presence of model uncertainty. Although robust approaches \cite{JankovicAutomatica18,AmesIEEEA20} have demonstrated success in this regard, in general, such techniques can be highly conservative. On the other hand, data-driven approaches have demonstrated the ability to reduce uncertainty and yield high-performance controllers in terms of both safety and stability. Popular data-driven approaches {\color{black} for} reducing uncertainty include work based on episodic learning \cite{TaylorL4DC20,SreenathADHS21} or by modeling the uncertainty using Gaussian processes (GPs) \cite{SreenathCDC21,DhimanTAC21}. However, providing strong guarantees in an episodic learning setting is challenging and, although {\color{black} GP-based approaches account for very general classes of uncertainties, GPs can be computationally intensive and the generality offered by GPs generally results in probabilistic, rather than deterministic, guarantees on stability and safety.}

As adaptive control \cite{Krstic} has a long history of success in controlling nonlinear systems with parametric uncertainty, there is also a rich line of work that unites CLFs and CBFs with techniques from adaptive control. The authors of \cite{TaylorACC20} extend the adaptive CLF (aCLF) paradigm \cite[Ch. 4.1]{Krstic} to CBFs, yielding the first instance of an \emph{adaptive} CBF\footnote{The term aCBF was also used in \cite{WeiTAC21-adaptive-cbf} to refer to a class of CBFs that account for time-varying control bounds.} (aCBF) that allows for the safe control of uncertain nonlinear systems with parametric uncertainty. The authors of \cite{LopezLCSS21,DixonACC21,SanfeliceACC21} extend the aCBF techniques from \cite{TaylorACC20} using set-membership identification, concurrent learning (CL) \cite{Chowdhary,DixonIJACSP19}, and hybrid techniques, respectively, which were shown to reduce the conservatism of original aCBF formulation. Nevertheless, all of the aforementioned aCBF techniques are limited to CBFs with relative degree one. In practice, however, many safety-critical constraints have relative degrees larger than one (e.g., constraints on the configuration of a mechanical system generally have at least relative degree two). The unification of CLFs/CBFs with techniques from CL adaptive control was also presented in \cite{AzimiCDC18,Azimi}; however, the resulting CLF controllers either only guarantee uniformly ultimately bounded stability or are limited to single-input feedback linearizable systems. Importantly, the CL-based aCBF controllers from \cite{AzimiCDC18,Azimi} do not provide strong safety guarantees since the CBF-based control inputs are generated using the estimated dynamics without accounting for estimation errors, leading to potential safety violations that can be understood through the notion of input-to-state-safety \cite{AmesLCSS19}.

To address high relative degree safety constraints for systems with \emph{known} dynamics, the authors of \cite{SreenathACC16,WeiCDC19,WeiTAC21-hocbf,DimosTAC21-hocbf} introduce exponential and high order CBFs (HOCBFs), which provide a systematic framework to construct CBFs that account for high relative degree constraints. Importantly, as noted in \cite{DimosTAC21-hocbf}, HOCBFs can also be used to simplify the search for valid CBFs since the dependence on higher order dynamics need not be directly encoded through the definition of the CBF itself (as in the relative degree one case). Rather, {\color{black} the} dependence {\color{black} on} higher order dynamics is \emph{implicitly} encoded through conditions on higher order derivatives of the CBF candidate. Despite the advancements of both aCBFs and HOCBFs, to our knowledge, the intersection of these two techniques has yet to be explored in the literature.

In this paper we unite aCBFs and HOCBFs to develop control policies satisfying high relative degree safety constraints for nonlinear systems with parametric uncertainty. Similar to \cite{DixonACC21}, our approach leverages the concurrent learning technique presented in \cite{DixonIJACSP19}, which identifies uncertain parameters of the nonlinear system online by exploiting sufficiently rich data collected along the system trajectory. A key insight enabling our high relative degree approach is that if the relative degrees of the CBF with respect to the control and uncertain parameters are the same (in a sense to be clarified later in this paper), then the sufficient conditions for safety can be encoded through affine constraints on the control input, allowing control synthesis to be performed in a computationally efficient quadratic programming framework. Furthermore, unlike existing aCBF formulations for relative degree one constraints \cite{TaylorACC20,LopezLCSS21,DixonACC21,SanfeliceACC21}, we show that our High Order Robust Adaptive Control Barrier Functions (HO-RaCBFs) inherit the robustness properties of zeroing CBFs \cite{AmesADHS15} in the sense that our developed aCBFs not only render the safe set forward invariant, but also \emph{asymptotically stable} when solutions begin outside the safe set. We then introduce a novel class of aCLF, termed exponentially stabilizing aCLFs (ES-aCLFs), that extends the CL paradigm from \cite{DixonIJACSP19} to a CLF setting by exploiting the same history stack used to reduce conservatism of the safety controller to endow a nominal CLF-based control policy with \emph{exponential} stability guarantees. The efficacy of the combined HO-RaCBF/ES-aCLF controller is demonstrated through simulations of a robotic navigation task and an inverted pendulum, both of which involve high relative degree safety constraints that cannot be addressed by existing aCBF approaches.

\section{Mathematical Preliminaries}
Consider a nonlinear control affine system of the form
\begin{equation}\label{eq:dyn}
    \dot{x}=f(x) + g(x)u,
\end{equation}
with state $x\in\mathbb{R}^n$ and control $u\in\mathcal{U}\subseteq\mathbb{R}^m$, where $f\,:\,\mathbb{R}^n\rightarrow\mathbb{R}^n$ and $g\,:\,\mathbb{R}^n\rightarrow\mathbb{R}^{n\times m}$ are locally Lipschitz vector fields modeling the drift and control directions, respectively. Given a feedback law $u=k(x,t)$, locally Lipschitz in $x$ and piecewise continuous in $t$, the closed-loop vector field $f_{\text{cl}}(x,t)\coloneqq f(x)+g(x)k(x,t)$ is also locally Lipschitz in $x$ and piecewise continuous in $t$, implying that \eqref{eq:dyn} admits a unique solution $x\,:\,\mathcal{I}\rightarrow\mathbb{R}^n$ starting from $x(0)\in\mathbb{R}^n$ on some maximal interval of existence $\mathcal{I}\subset\mathbb{R}_{\geq0}$. A closed set $\mathcal{C}\subset\mathbb{R}^n$ is said to be \emph{forward invariant} for the closed-loop system $\dot{x}=f_{\text{cl}}(x,t)$ if $x(0)\in\mathcal{C}\implies x(t)\in\mathcal{C}$ for all $t\in\mathcal{I}$. In this paper (and in the related literature) forward invariance is used to formalize the abstract notion of safety. Hence, if a given ``safe" set $\mathcal{C}$ is forward invariant for $\dot{x}=f_{\text{cl}}(x,t)$, then we say the closed-loop system is \emph{safe} with respect to $\mathcal{C}$. It will be assumed throughout this paper that any safe set $\mathcal{C}$ can be expressed as the zero-superlevel set of a continuously differentiable function $h\,:\,\mathbb{R}^n\rightarrow\mathbb{R}$ as
\begin{equation}\label{eq:C}
    \mathcal{C}=\{x\in\mathbb{R}^n\,|\,h(x)\geq0\}.
\end{equation}
A popular tool for developing controllers that render \eqref{eq:C} forward invariant for \eqref{eq:dyn} is the concept of a CBF \cite{AmesTAC17,AmesECC19}, which places Lyapunov-like conditions on the derivative of $h$ to guarantee safety. A limitation of traditional CBFs from \cite{AmesTAC17}, however, is that their effectiveness is conditioned upon the assumption that the function $h$ has relative degree one. Yet, many relevant systems and safe sets fail to satisfy such a condition, which has motivated the introduction of exponential and higher order CBFs \cite{SreenathACC16,WeiCDC19,WeiTAC21-hocbf,DimosTAC21-hocbf} to account for safety constraints with \emph{high relative degree}. Before proceeding, we recall that the \emph{Lie derivative} of a differentiable function $h\,:\,\mathbb{R}^n\rightarrow\mathbb{R}$ along a vector field $f\,:\,\mathbb{R}^n\rightarrow\mathbb{R}^n$ is defined as $L_fh(x)\coloneqq \tfrac{\partial h}{\partial x}f(x)$. This notation allows us to denote higher order Lie derivatives along an additional vector field $g$ as $L_gL_f^{i-1}h(x)=\tfrac{\partial(L_f^{i-1}h)}{\partial x}g(x)$ (see e.g. \cite[Ch. 13.2]{Khalil}). 

\begin{definition}[\cite{Khalil}]
    A sufficiently smooth function $h\,:\,\mathbb{R}^n\rightarrow\mathbb{R}$ is said to have \emph{relative degree} $r\in\mathbb{N}$ with respect to \eqref{eq:dyn} on a set $\mathcal{R}\subset\mathbb{R}^n$ if 1) for all $1\leq i\leq r-1$, $L_gL_f^{i-1}h(x)\equiv0$; 2) $L_gL_f^{r-1}h(x)\neq0$ for all $x\in\mathcal{R}$.
\end{definition}

To account for high relative degree safety constraints, the authors of \cite{WeiCDC19,WeiTAC21-hocbf,DimosTAC21-hocbf} introduce the notion of a HOCBF. Before stating the definition, we recall from \cite{AmesTAC17} that a continuous function $\alpha\,:\,(-b,a)\rightarrow(-\infty,\infty)$, for some $a,b\in\mathbb{R}_{>0}$, is said to be an \emph{extended class} $\mathcal{K}$ \emph{function} if it is strictly increasing and $\alpha(0)=0$.

\begin{definition}[\cite{WeiCDC19,WeiTAC21-hocbf,DimosTAC21-hocbf}]\label{def:HOCBF}
    Consider system \eqref{eq:dyn} and a set $\mathcal{C}\subset\mathbb{R}^n$ as in \eqref{eq:C}. Let  $\{\mathcal{C}_i\}_{i=1}^r$ be a collection of sets of the form $\mathcal{C}_i\coloneqq\{x\in\mathbb{R}^n\,|\,\psi_{i-1}(x)\geq0\}$, where $\psi_{0}(x)\coloneqq h(x)$ and
    \begin{equation}\label{eq:psi}
        \begin{aligned}
        \psi_{i}(x)\coloneqq & \dot{\psi}_{i-1}(x) + \alpha_{i}(\psi_{i-1}(x)),\; i\in\{1,\dots,r-1\},\\
        \psi_{r}(x,u)\coloneqq & \dot{\psi}_{r-1}(x,u) + \alpha_r(\psi_{r-1}(x)),
        \end{aligned}
    \end{equation}
    where $\{\alpha_{i}\}_{i=1}^r$ is a collection of differentiable extended class $\mathcal{K}$ functions. Then, the function $h$ is said to be a HOCBF of order $r$ for \eqref{eq:dyn} on an open set $\mathcal{D}\supset\cap_{i=1}^r\mathcal{C}_i$ if $h$ has relative degree $r$ on some nonempty $\mathcal{R}\subseteq\mathcal{D}$ and there exists a suitable choice of $\{\alpha_{i}\}_{i=1}^r$ such that for all $x\in\mathcal{D}$
    \[
        \sup_{u\in\mathcal{U}}\{\underbrace{L_{f}\psi_{r-1}(x) + L_g\psi_{r-1}(x)u + \alpha_r(\psi_{r-1}(x))}_{\psi_r(x,u)}\} \geq0.
    \]
\end{definition}

If $\mathcal{U}=\mathbb{R}^m$, the above states that $h$ is a HOCBF if $\|L_g\psi_{r-1}h(x)\|=0$ $\implies$ $L_f\psi_{r-1}h(x) \geq - \alpha_r(\psi_{r-1}(x))$, implying that $h$ need not have \emph{uniform} relative degree on $\mathcal{D}$ as illustrated in \cite{DimosTAC21-hocbf}, provided the unforced dynamics satisfy the above condition at points where $\|L_g\psi_{r-1}h(x)\|=0$. The following result provides higher order conditions for safety.

\begin{theorem}[\cite{DimosTAC21-hocbf}]\label{theorem:HOCBF}
    Let $h$ be a HOCBF for \eqref{eq:dyn} on $\mathcal{D}\subset\mathbb{R}^n$ as in Def. \ref{def:HOCBF}. Then, any locally Lipschitz controller $u=k(x)\in K_{\text{cbf}}(x)$, where $K_{\text{cbf}}(x)\coloneqq  \{u\in\mathcal{U}\,|\,\psi_r(x,u)\geq0\}$, renders $\cap_{i=1}^r\mathcal{C}_i^r$ forward invariant for the closed-loop system. 
\end{theorem}

\section{High Order Robust Adaptive Control Barrier Functions}
This section introduces the concept of a High Order Robust Adaptive Control Barrier Function (HO-RaCBF), which provides a tool to synthesize controllers that guarantee the satisfaction of high relative degree safety constraints for nonlinear systems with parametric uncertainty.
To this end, we now turn our attention to systems of the form
\begin{equation}\label{eq:dyn2}
    \dot{x}=f(x) + Y(x)\theta + g(x)u,
\end{equation}
where $f$ and $g$ are known locally Lipschitz vector fields as in \eqref{eq:dyn}, $Y\,:\,\mathbb{R}^n\rightarrow\mathbb{R}^{n\times p}$ is a known locally Lipschitz regression matrix, and $\theta\in\mathbb{R}^p$ is a constant vector of uncertain parameters. We assume that $f(0)=0$ and $Y(0)=0$ so that 0 is an equilibrium point of the unforced system. Our main objective is to synthesize controllers for \eqref{eq:dyn2} that guarantee safety under the presumption that the function $h$ defining the safe set $\mathcal{C}$ as in \eqref{eq:C} has a high relative degree with respect to \eqref{eq:dyn2}. To facilitate our approach, we make the following assumption on the structure of the uncertainty in \eqref{eq:dyn2}.

\begin{assumption}\label{assumption:rd}
Consider a set $\mathcal{C}$ as in \eqref{eq:C} and an open set $\mathcal{D}$ as in Def. \ref{def:HOCBF}. If $h$ has relative degree $r$ on $\mathcal{R}\subseteq\mathcal{D}$ with respect to \eqref{eq:dyn2} (i.e., if there exist some nonempty $\mathcal{R}\subseteq\mathcal{D}$ such that $L_gL_f^{i-1}h(x)\equiv 0$ for all $1\leq i\leq r-1$ and $L_gL_f^{r-1}h(x)\neq0$ for all $x\in\mathcal{R}$), then there exists $\mathcal{R}'\subseteq\mathcal{D}$ such that $L_YL_f^{i-1}h(x)\equiv 0$ for all $1\leq i\leq r-1$ and $L_YL_f^{r-1}h(x)\neq0$ for all $x\in\mathcal{R}'$.
\end{assumption}

\begin{remark}
    The above assumption requires that the uncertainty in \eqref{eq:dyn2} does not appear before the control when taking higher order derivatives of $h$. Although this may seem restrictive, a variety of physical systems satisfy Assumption \ref{assumption:rd}. Examples include Lagrangian mechanical systems, where $h$ is a function of only the system's configuration. If the above assumption is not made, then the uncertain parameters $\theta$ will appear alongside the control input $u$ in higher order terms, complicating the formulation of the affine constraints on $u$ developed in this paper. {\color{black} From an adaptive control perspective, Assumption \ref{assumption:rd} is similar to the assumption that the uncertain parameters satisfy the matching condition.}
\end{remark}

According to Theorem \ref{theorem:HOCBF}, the forward invariance of $\mathcal{C}$ can be enforced by ensuring that the control input is selected such that the HOCBF condition from Def. \ref{def:HOCBF} is satisfied for all $x\in\cap_{i=1}^{r}\mathcal{C}_i$; however, the presence of model uncertainty in \eqref{eq:dyn2} makes it impossible to directly enforce such a condition. To address this challenge, we aim to take a data-driven approach and update the estimates of the uncertain parameters online using techniques from adaptive control \cite{Krstic,Chowdhary,DixonIJACSP19,DixonACC21} while guaranteeing safety at all times. Following the approach from \cite{DixonIJACSP19}, observe that integrating \eqref{eq:dyn2} over a time interval $[t-\Delta T, t]\subset\mathbb{R}$ using the Fundamental Theorem of Calculus allows \eqref{eq:dyn2} to be equivalently represented as 
\begin{equation}\label{eq:intdyn}
    \Delta x(t)=x(t)-x(t-\Delta T)=\mathcal{F}(t) + \mathcal{Y}(t)\theta + \mathcal{G}(t),
\end{equation}
where $\mathcal{F}(t)\coloneqq \int_{t-\Delta T}^{t}f(x(\tau))d\tau$, $\mathcal{Y}(t)\coloneqq \int_{t-\Delta T}^{t}Y(x(\tau))d\tau$, $\mathcal{G}(t)\coloneqq \int_{t-\Delta T}^{t}g(x(\tau))u(\tau)d\tau$. Now let $\mathcal{H}\coloneqq\{(t_j,x_j,x_{j}^{-},u_j) \}_{j=1}^M$ be a history stack of $M\in\mathbb{N}$ instances of input-output data, where $t_j\in[\Delta T, t]$ denotes a sampling time, $x_j\coloneqq x(t_j)$, $x_j^{-}\coloneqq x(t_j-\Delta t)$, $u_j\coloneqq u(t_j)$, and define\footnote{The function $\Lambda$ is implicitly a function of time as data is added/removed from the history stack $\mathcal{H}$ along the system trajectory.}
\begin{equation}\label{eq:Lambda}
\Lambda(t)\coloneqq\sum_{j=1}^{M}\mathcal{Y}_j^\top\mathcal{Y}_j,\quad \lambda(t)\coloneqq \lambda_{\min}(\Lambda(t)),
\end{equation}
where $\lambda_{\min}(\Lambda)$ denotes the minimum eigenvalue of $\Lambda$. As noted in \cite{DixonIJACSP19,DixonACC21}, the function $\lambda$ is piecewise constant between sampling times and nonnegative since $\Lambda(t)$ is at least positive semidefinite at all times. The following assumption will be used to ensure safety for all possible realizations of the uncertain parameters.

{\color{black}
\begin{assumption}\label{assumption:estimation_error}
    The uncertain parameters $\theta$ belong to a known convex polytope $\Theta\subset\R^p$.
\end{assumption}
The above assumption implies that for any given parameter estimate $\hat{\theta}\in\Theta$ there exists some maximum possible estimation error $\tilde{\vartheta}\in\R^p$ in the sense that $\|\theta-\hat{\theta}\|\leq\|\tilde{\vartheta}\|$ for all $\hat{\theta}\in\Theta$. Given that $\Theta$ is a convex polytope, each component of $\tilde{\vartheta}$ can be computed as 
\begin{equation*}
    \tilde{\vartheta}_i=\max\big\{\big\vert\min_{\theta,\hat{\theta}\in\Theta}\theta_i - \hat{\theta}_i\big\vert,\, \big\vert\max_{\theta,\hat{\theta}\in\Theta}\theta_i - \hat{\theta}_i\big\vert\big\},
\end{equation*}
where $|\cdot|$ denotes aboslute value and $\theta_i$ denotes the $i$th component of $\theta$, which requires solving a pair of linear programs for each parameter.} The following lemma, adapted from \cite{DixonACC21}, provides a verifiable bound on the parameter estimation error.
\begin{lemma}[\cite{DixonACC21}]\label{lemma:theta_bound}
Consider system \eqref{eq:dyn2} and suppose the estimated parameters are updated according to
\begin{equation}\label{eq:theta_dot_cbf}
    \dot{\hat{\theta}}=\gamma\sum_{j= 1}^M\mathcal{Y}_j^\top (\Delta x_j - \mathcal{F}_j - \mathcal{Y}_j\hat{\theta} - \mathcal{G}_j),
\end{equation}
where $\Delta x_j\coloneqq x_j-x_j^-$, $\mathcal{Y}_j\coloneqq \mathcal{Y}(t_j)$, $\mathcal{F}_j\coloneqq \mathcal{F}(t_j)$, $\mathcal{G}_j\coloneqq \mathcal{G}(t_j)$, and $\gamma\in\mathbb{R}_{>0}$ is an adaptation gain. Provided Assumption \ref{assumption:estimation_error} holds and $\hat{\theta}(0)\in\Theta$, then the parameter estimation error $\tilde{\theta}=\theta-\hat{\theta}$ is bounded for all $t\in\mathcal{I}$ as
\begin{equation}\label{eq:theta_bound}
    \|\tilde{\theta}(t)\|\leq\nu(t)\coloneqq \|\tilde{\vartheta}\|e^{-\gamma\int_{0}^t\lambda(\tau)d\tau}.
\end{equation}
\end{lemma}

The above lemma implies that, under the update law in \eqref{eq:theta_dot_cbf}, the parameter estimation error is always bounded by a known value provided the initial parameter estimates are selected such that $\hat{\theta}(0)\in\Theta$. Moreover, if there exists some time $T$ such that $\lambda(t)>0$ for all $t>T$, then the bound in \eqref{eq:theta_bound} implies all estimated parameters exponentially converge to their true values\footnote{See \cite[Ch. 3]{Chowdhary} for a discussion on the relation between traditional persistence of excitation conditions and the milder \emph{finite} excitation conditions leveraged in concurrent learning adaptive control required to achieve $\lambda>0$. We also refer the reader to \cite{ChowdharyACC11} for various algorithms that record data points in $\mathcal{H}$ so as to ensure $\lambda(t)$ is always nondecreasing.\label{footnote:data}}. We now have the {\color{black}necessary} tools in place to introduce a new class of aCBF that allows for the consideration of high relative degree safety constraints.

\begin{definition}\label{def:HO-RaCBF}
    Consider system \eqref{eq:dyn2} and a safe set $\mathcal{C}\subset\mathbb{R}^n$ as in \eqref{eq:C}. Consider a collection of sets $\{\mathcal{C}_i\}_{i=1}^r$ of the form $\mathcal{C}_i\coloneqq\{x\in\mathbb{R}^n\,|\,\psi_{i-1}(x)\geq0\}$, where $\psi_{0}(x)\coloneqq h(x)$ and $\{\psi_{i}\}_{i=1}^r$ are defined as in \eqref{eq:psi}. The sufficiently smooth function $h$ is said to be a \emph{high order robust adaptive control barrier function} (HO-RaCBF) of order $r$ for \eqref{eq:dyn2} on an open set $\mathcal{D}\supset\cap_{i=1}^r\mathcal{C}_i$ if $h$ has relative degree $r$ on some nonempty $\mathcal{R}\subseteq\mathcal{D}$ and there exists a suitable choice of $\{\alpha_{i}\}_{i=1}^r$ as in \eqref{eq:psi} such that for all $x\in\mathcal{D}$, $\theta\in\Theta$, and $t\in\mathcal{I}$
    \begin{equation}\label{eq:HO-RaCBF}
    \begin{aligned}
        \sup_{u\in\mathcal{U}}\left\{L_{f}\psi_{r-1}(x) + L_Y\psi_{r-1}(x)\theta + L_g\psi_{r-1}(x)u \right\} \\
        \geq -\alpha_r(\psi_{r-1}(x)) + \|L_Y\psi_{r-1}(x)\|\nu(t),
    \end{aligned}
    \end{equation}
    where $\nu$ is defined as in \eqref{eq:theta_bound}.
\end{definition}

Intuitively, the above condition adds a buffer to the original HOCBF condition from Def. \ref{def:HOCBF} to account for all possible realizations of the uncertain parameters given the set $\Theta$. This buffer may shrink over time as the uncertain parameters are identified and exponentially converges to zero in the limit as $t\rightarrow\infty$ provided there exists a time $T$ for which $\lambda(t)>0$ for all $t\geq T$. Furthermore, Def. \ref{def:HO-RaCBF} allows us to consider the set of all control values satisfying \eqref{eq:HO-RaCBF} as 
\begin{equation}\label{eq:Kcbf}
\begin{aligned}
    \hat{K}_{\text{cbf}}(x,\theta,t)\coloneqq & \{u\in\mathcal{U}\,|\, L_{f}\psi_{r-1}(x) + L_Y\psi_{r-1}(x)\theta \\ & + L_g\psi_{r-1}(x)u  +\alpha_r(\psi_{r-1}(x)) \\ &- \|L_Y\psi_{r-1}(x)\|\nu(t) \geq0\}.
\end{aligned}
\end{equation}
The following theorem shows that any well-posed control law $u=k(x,\hat{\theta},t)$ satisfying $k(x,\hat{\theta},t)\in \hat{K}_{\text{cbf}}(x,\hat{\theta},t)$ renders $\cap_{i=1}^r\mathcal{C}_i$ forward invariant for \eqref{eq:dyn2}.

\begin{theorem}\label{theorem:HO-RaCBF}
    Consider system \eqref{eq:dyn2}, a set $\mathcal{C}$ defined by a sufficiently smooth function $h$ as in \eqref{eq:C}, and let $h$ be a HO-RaCBF on $\mathcal{D}$. Provided Assumptions \ref{assumption:rd}-\ref{assumption:estimation_error} hold and the estimated parameters are updated according to \eqref{eq:theta_dot_cbf}, then any controller $u=k(x,\hat{\theta},t)$ locally Lipschitz in $(x,\hat{\theta})$ and piecewise continuous in $t$ satisfying $k(x,\hat{\theta},t)\in \hat{K}_{\text{cbf}}(x,\hat{\theta},t)$ renders $\cap_{i=1}^r\mathcal{C}_i$ forward invariant for \eqref{eq:dyn2}. 
\end{theorem}

\begin{proof}
    According to Theorem \ref{theorem:HOCBF}, to guarantee forward invariance of $\cap_{i=1}^r\mathcal{C}_i$, it is sufficient to show that for each $x\in\cap_{i=1}^r\mathcal{C}_i$ the input is selected such as $\psi_r(x,u)\geq0$. For the dynamics in \eqref{eq:dyn2}, if Assumption \ref{assumption:rd} holds, then $\psi_r(x,u)=L_{f}\psi_{r-1}(x) + L_{Y}\psi_{r-1}(x)\theta + L_g\psi_{r-1}(x)u + \alpha_r(\psi_{r-1}(x))$. Hence, it is our aim to show that any $u=k(x,\hat{\theta},t)$ satisfying $k(x,\hat{\theta},t)\in \hat{K}_{\text{cbf}}(x,\hat{\theta},t)$ satisfies $\psi_r(x,k(x,\hat{\theta},t))\geq0$ for all $x\in\cap_{i=1}^r\mathcal{C}_i$, all $\hat{\theta}\in\Theta$, and all $t\in\mathcal{I}$. To this end, observe that under control $u=k(x,\hat{\theta},t)\in \hat{K}_{\text{cbf}}(x,\hat{\theta},t)$ (omitting functional arguments for ease of presentation)
    \[
    \begin{aligned}
    \psi_r
    =& L_{f}\psi_{r-1} + L_{Y}\psi_{r-1}\hat{\theta} + L_{Y}\psi_{r-1}\tilde{\theta} \\ & +  L_g\psi_{r-1}k + \alpha_r(\psi_{r-1}) \\
    \geq & L_{f}\psi_{r-1} + L_{Y}\psi_{r-1}\hat{\theta}  +  L_g\psi_{r-1}k \\ &+ \alpha_r(\psi_{r-1}) - \|L_Y\psi_{r-1}\|\|\tilde{\theta}\| \\
    \geq & L_{f}\psi_{r-1} + L_{Y}\psi_{r-1}\hat{\theta}  +  L_g\psi_{r-1}k \\ &+ \alpha_r(\psi_{r-1}) - \|L_Y\psi_{r-1}\|\nu \\ 
    \geq &  0, 
    \end{aligned}
    \]
    for all $x\in\cap_{i=1}^r\mathcal{C}_i$, all $\hat{\theta}\in\Theta$, and all $t\in\mathcal{I}$.
    In the above, the first inequality follows from the fact that $L_{Y}\psi_{r-1}(x)\tilde{\theta}\geq - \|L_{Y}\psi_{r-1}(x)\|\|\tilde{\theta}\|$, the second from the bound in \eqref{eq:theta_bound}, and the third from \eqref{eq:Kcbf}. Since $\psi_r(x,k(x,\hat{\theta},t))\geq0$ holds for all $x\in\cap_{i=1}^r\mathcal{C}_i$, all $\hat{\theta}\in\Theta$, and all $t\in\mathcal{I}$, it follows from Theorem \ref{theorem:HOCBF} that $\cap_{i=1}^r\mathcal{C}_i$ is forward invariant for the closed-loop system, as desired.
\end{proof}

Definition \ref{def:HO-RaCBF} and Theorem \ref{theorem:HO-RaCBF} generalize the ideas introduced in \cite{DixonACC21} to constraints with high relative degree, thereby facilitating the application of such ideas to more complex systems and safe sets. Although the results of \cite{DixonACC21} apply to safe sets defined by multiple barrier functions, whereas ours apply only to those defined by a single barrier function, there exist various approaches in the CBF literature \cite{LarsLCSS19,EgerstedtLCSS17} to formally combine multiple barrier functions\footnote{In practice, it is common to simply include multiple CBF-based constraints in a quadratic program such as the one proposed in \eqref{eq:HO-RaCBF-QP}, which, as we demonstrate empirically in Sec. \ref{sec:sim}, allows one to consider safe sets defined by multiple barrier functions.} using smooth approximations of min/max operators \cite{LarsLCSS19} or nonsmooth analysis \cite{EgerstedtLCSS17}. We now show that if the conditions of Theorem \ref{theorem:HO-RaCBF} hold on $\mathcal{D}$ and $x(0)\in\mathcal{D}\backslash\cap_{i=1}^r\mathcal{C}_i$, then any controller satisfying $k(x,\hat{\theta},t)\in \hat{K}_{\text{cbf}}(x,\hat{\theta},t)$ also guarantees asymptotic stability of $\mathcal{C}$, which ensures the developed controller is robust to perturbations \cite{AmesADHS15}.

\begin{corollary}\label{corollary:HO-RaCBF}
    Let the conditions of Theorem \ref{theorem:HO-RaCBF} hold and suppose that $x(0)\in\mathcal{D}\backslash\cap_{i=1}^r\mathcal{C}_i$. Provided $u=k(x,\hat{\theta},t)\in \hat{K}_{\text{cbf}}(x,\hat{\theta},t)$ renders the closed-loop dynamics \eqref{eq:dyn2} forward complete, then the set $\cap_{i=1}^r\mathcal{C}_i$ is asymptotically stable for the closed-loop system.
\end{corollary}

\begin{proof}
    It was shown in Theorem \ref{theorem:HO-RaCBF} that the proposed controller $k(x,\hat{\theta},t)\in \hat{K}_{\text{cbf}}(x,\hat{\theta},t)$ satisfies $\psi_r(x,k(x,\hat{\theta},t))\geq0$ for all $x\in\mathcal{D}$, $\hat{\theta}\in\Theta$ and $t\in\mathcal{I}$. Hence, by definition, $k(x,\hat{\theta},t)\in K_{\text{cbf}}(x)$ for all $x\in\mathcal{D}$, $\hat{\theta}\in\Theta$, and $t\in\mathcal{I}$. It then follows from \cite[Prop. 3 and Rem. 2]{DimosTAC21-hocbf} that $\cap_{i=1}^r\mathcal{C}_i$ is asymptotically stable for the closed-loop system.
\end{proof}

A controller satisfying the conditions of Theorem \ref{theorem:HO-RaCBF} can be computed through the use of quadratic programming (QP). Specifically, given an estimate of the uncertain parameters $\hat{\theta}$ and  a nominal feedback control policy $k_d(x,\hat{\theta},t)$, a minimally invasive safe controller can be computed using the following HO-RaCBF-QP
\begin{equation}\label{eq:HO-RaCBF-QP}
    \begin{aligned}
    \min_{u\in\mathcal{U}} &\quad \tfrac{1}{2}\|u-k_d(x,\hat{\theta},t)\|^2 \\ 
     \text{s.t.} & \quad L_{f}\psi_{r-1}(x) + L_Y\psi_{r-1}(x)\hat{\theta} + L_{g}\psi_{r-1}(x)u \\ & \geq -\alpha_{r}(\psi_{r-1}(x)) + \|L_Y\psi_{r-1}(x)\|\nu(t),
    \end{aligned}
\end{equation}
which enforces the conditions of Theorem \ref{theorem:HO-RaCBF} provided the resulting controller is Lipschitz continuous and $h$ is a valid HO-RaCBF with $u\in\mathcal{U}$. That is, the above QP \eqref{eq:HO-RaCBF-QP} allows the nominal policy $k_d$ to be executed on \eqref{eq:dyn2} if $k_d$ can be formally verified as safe and intervenes in a minimally invasive fashion to guarantee safety only if $k_d$ cannot be certified as safe. We illustrate in the following section how one can leverage the same history stack used to reduce uncertainty in the parameter estimates to synthesize a desired policy $k_d$ with exponential stability guarantees. 

\section{Exponentially Stabilizing Adaptive Control Lyapunov Functions}
In this section we introduce the concept of an exponentially stabilizing adaptive control Lyapunov function (ES-aCLF) as a tool to exponentially stabilize uncertain nonlinear systems in the presence of parametric uncertainty.

\begin{definition}\label{def:ES-aCLF}
A continuously differentiable positive definite function $V\,:\,\mathbb{R}^n\rightarrow\mathbb{R}_{\geq0}$ is said to be an \emph{exponentially stabilizing adaptive control Lyapunov function} (ES-aCLF) for \eqref{eq:dyn2} if there exist positive constants $c_1,c_2,c_3\in\mathbb{R}_{>0}$ such that for all $x\in\mathbb{R}^n$ and $\theta\in\mathbb{R}^p$
\begin{subequations}
    \begin{equation}
        c_1\|x\|^2\leq V(x)\leq c_2\|x\|^2,
    \end{equation}
    \begin{equation}\label{eq:aclf}
        \inf_{u\in\mathcal{U}}\{L_{f}V(x) + L_YV(x)\theta + L_gV(x)u \}\leq -c_3V(x).
    \end{equation}
\end{subequations}
\end{definition}
Given the above definition, let 
\begin{equation}\label{eq:Kclf}
\begin{aligned}
    K_{\text{clf}}(x,\theta)\coloneqq & \{u\in\mathcal{U}\,|\ L_{f}V(x)  \\ &+  L_YV(x)\theta + L_gV(x)u\leq -c_3V(x)\}, 
\end{aligned}
\end{equation}
denote the point-wise set of all control values satisfying \eqref{eq:aclf}. The following lemma provides a parameter update law that can be combined with any locally Lipschitz control policy satisfying $k(x,\hat{\theta})\in K_{\text{clf}}(x,\hat{\theta})$ to guarantee stability. 

\begin{lemma}\label{lemma:aclf}
    Consider system \eqref{eq:dyn2}. Let $V$ be an ES-aCLF as in Def. \ref{def:ES-aCLF} and define $z\coloneqq\begin{bmatrix} x^\top & \tilde{\theta}^\top \end{bmatrix}^\top$. Provided the estimates of the unknown parameters are updated according to
    \begin{equation}\label{eq:theta_dot}
        \dot{\hat{\theta}}=\Gamma L_YV(x)^\top + \gamma\Gamma\sum_{j=1}^M\mathcal{Y}_j^\top(\Delta x_j - \mathcal{F}_j -  \mathcal{Y}_j\hat{\theta} - \mathcal{G}_j),
    \end{equation}
    where $\Gamma\in\mathbb{R}^{p\times p}$ is a positive definite gain matrix and $\gamma\in\mathbb{R}_{>0}$ is a user-defined adaptation gain, then any locally Lipschitz controller $u=k(x,\hat{\theta})$ satisfying $k(x,\hat{\theta})\in K_{clf}(x,\hat{\theta})$
    ensures that the composite system trajectory $t\mapsto z(t)$ remains bounded in the sense that for all $t\in[0,\infty)$
    \begin{equation}\label{eq:z_bound}
        \|z(t)\|\leq\sqrt{\tfrac{\eta_2}{\eta_1}}\|z(0)\|,
    \end{equation}
    where $\eta_1\coloneqq\min\{c_1,\tfrac{1}{2}\lambda_{\min}(\Gamma^{-1})\}$ and $\,\eta_2\coloneqq\max\{c_2,\tfrac{1}{2}\lambda_{\max}(\Gamma^{-1})\}$ are positive constants. Moreover
    \[
    \lim_{t\rightarrow\infty}x(t)=0. 
    \]
\end{lemma}

\begin{proof}
    Consider the Lyapunov function candidate $V_a(z)\coloneqq V(x) + \frac{1}{2}\tilde{\theta}^\top\Gamma^{-1}\tilde{\theta}$, which can be bounded for all $z\in\mathbb{R}^{n+p}$ as $\eta_1\|z\|^2\leq V_a(z)\leq\eta_2\|z\|^2$. Taking the derivative of $V_a$ along the composite system trajectory yields
    \begin{equation}\label{eq:Vadot}
    \begin{aligned}
    \dot{V}_a= & L_fV(x) + L_YV(x)\theta   + L_gV(x)u - \tilde{\theta}^\top L_YV(x)^\top  \\
    & - \gamma\tilde{\theta}^\top\sum_{j=1}^M\mathcal{Y}_j^\top(\Delta x_j - \mathcal{F}_j -  \mathcal{Y}_j\hat{\theta} - \mathcal{G}_j)\\
    = & L_fV(x) + L_YV(x)\hat{\theta} + L_gV(x)u -\gamma\tilde{\theta}^\top\Lambda(t)\tilde{\theta},
    \end{aligned}
    \end{equation}
    where $\Lambda$ is from \eqref{eq:Lambda}. Using the fact that $\Lambda(t)$ is at least positive semi-definite for all time implies $\dot{V}_a$ can be bounded as $\dot{V}_a\leq L_fV(x) + L_YV(x)\hat{\theta} + L_gV(x)u$. Choosing $u=k(x,\hat{\theta})\in K_{\text{clf}}(x,\hat{\theta})$ and the hypothesis that $V$ is a valid ES-aCLF allows $\dot{V}_a$ to be further bounded as $\dot{V}_a\leq-c_3V(x)\leq0$, revealing that $\dot{V}_a$ is negative semi-definite. Hence, $V_a$ is nonincreasing and $V_a(z(t))\leq V_a(z(0))$ for all $t\in[0,\infty)$, which can be combined with the bounds on $V_a$ to yield \eqref{eq:z_bound}. Since $V$ is continuous and $\dot{V}_a\leq-c_3V(x)\leq0$, it follows from the LaSalle-Yoshizawa theorem \cite[Thm. A.8]{Krstic} that $\lim_{t\rightarrow\infty}c_3V(x)=0$, implying $\lim_{t\rightarrow\infty}x(t)=0$.
\end{proof}

The following theorem shows that if sufficiently rich data is collected along the system trajectory, then $x(t)$ and $\tilde{\theta}(t)$ both exponentially converge to the origin.

\begin{theorem}\label{theorem:ES-aCLF}
Under the assumption that the conditions of Lemma \ref{lemma:aclf} hold, suppose that there exists a time $T\in\mathbb{R}_{\geq0}$ and a positive constant $\underline{\lambda}\in\mathbb{R}_{>0}$ such that $\lambda(t)\geq\underline{\lambda}$ for all $t\in[T,\infty)$. Then, for all $t\in[0,T)$, $z(t)$ is bounded in the sense that \eqref{eq:z_bound} holds. Furthermore, for all $t\in[T,\infty)$, $z(t)$ exponentially converges to the origin at a rate proportional to $\eta_3\coloneqq \min\{\gamma\underline{\lambda},c_1c_3\}$ in the sense that for all $t\in[T,\infty)$
\begin{equation}\label{eq:z_exp_bound}
    \|z(t)\|\leq\sqrt{\tfrac{\eta_2}{\eta_1}}\|z(T)\|e^{-\frac{\eta_3}{2\eta_2}(t-T)}.
\end{equation}
\end{theorem}

\begin{proof}
    Since $\lambda(t)\geq0$ for all $t\in[0,T)$ the conclusions of Lemma \ref{lemma:aclf} hold, implying $z(t)$ is bounded as in \eqref{eq:z_bound} for all $t\in[0,T)$. Provided $\lambda(t)\geq\underline{\lambda}$ for all $t\in[T,\infty)$ and $u=k(x,\hat{\theta})\in K_{clf}(x,\hat{\theta})$ for all $(x,\hat{\theta})\in\mathbb{R}^{n+p}$, then \eqref{eq:Vadot} can be bounded as
    \[
        \dot{V}_a\leq -c_3V(x) - \gamma\underline{\lambda}\|\tilde{\theta}\|^2\leq -\eta_3\|z\|^2\leq -\frac{\eta_3}{\eta_2}V_a.
    \]
    Invoking the comparison lemma \cite[Lem. 3.4]{Khalil} implies $t\mapsto V_a(z(t))$ is bounded for all $t\in[T,\infty)$ as 
    \[
        V_a(z(t))\leq V_a(z(T))e^{-\frac{\eta_3}{\eta_2}(t-T)},
    \]
    which can be combined with the bounds on $V_a$ to yield \eqref{eq:z_exp_bound}.
\end{proof}

Similar to the previous section, given an estimate of the uncertain parameters $\hat{\theta}$, control inputs satisfying \eqref{eq:aclf} can be computed using the following ES-aCLF-QP
\begin{equation}\label{eq:ES-aCLF-QP}
    \begin{aligned}
    \min_{u\in\mathcal{U}}\quad & \tfrac{1}{2}u^\top u \\
    \text{s.t.}\quad & L_fV(x) + L_YV(x)\hat{\theta} + L_gV(x)u\leq -c_3V(x).
    \end{aligned}
\end{equation}

A controller guaranteeing stability and safety\footnote{Similar to works such as \cite{TaylorACC20,LopezLCSS21}, the parameter update laws for the proposed adaptive CBF and CLF are different | if one wishes to combine the two in a single QP-based controller, separate estimates of the uncertain parameters must be maintained. Despite this, note that the data from a single history stack can be used in both update laws.} can then be synthesized by either taking the solution to \eqref{eq:ES-aCLF-QP} as $k_d$ in \eqref{eq:HO-RaCBF-QP} or by forming a single QP with the HO-RaCBF constraint from \eqref{eq:HO-RaCBF-QP} and a relaxed version of the ES-aCLF constraint from \eqref{eq:ES-aCLF-QP} to guarantee feasibility {\color{black} provided the resulting controller is Lipschitz continuous}.

\begin{remark}
    Note that asymptotic stability of the origin for the closed-loop system \eqref{eq:dyn2} with $u=k(x,\hat{\theta})\in K_{\text{clf}}(x,\hat{\theta})$ is guaranteed by Lemma \ref{lemma:aclf} regardless of whether or not the richness of data condition $\lambda(t)\geq\underline{\lambda}>0$ is satisfied. In this situation, the stability guarantees induced by the ES-aCLF reduce to those of the classically defined adaptive CLF from \cite[Ch. 4.1]{Krstic} (see also \cite{TaylorACC20}). In this regard, neither safety nor stability is predicated upon collecting sufficiently rich data | this data is exploited only to reduce conservatism of the HO-RaCBF controller and to endow the ES-aCLF controller with exponential convergence guarantees.
\end{remark}
    
\begin{remark}
    The concept of an ES-aCLF generalizes the adaptive control designs from \cite{DixonIJACSP19} in that the Lyapunov functions used to verify stability of the controllers proposed in \cite{DixonIJACSP19} meet the criteria of an ES-aCLF posed in Def. \ref{def:ES-aCLF}. We refer the reader to works such as \cite{AmesTAC14} for a discussion on the potential advantages of using an optimization-based CLF-based control law as in \eqref{eq:ES-aCLF-QP} over a traditional closed-form feedback control law such as those posed in \cite{DixonIJACSP19}.
\end{remark}

\section{Case Studies}\label{sec:sim}
\subsection{Robotic Navigation}
Our first example involves a robotic navigation task for a system modeled as a double integrator with uncertain friction effects. The system is in the form of \eqref{eq:dyn2}:
\[
\begin{aligned}
\underbrace{
\begin{bmatrix}
\dot{x}_1 \\ \dot{x}_2 \\ \dot{x}_3 \\ \dot{x}_4
\end{bmatrix}
}_{\dot{x}}
=
\underbrace{
\begin{bmatrix}
x_3 \\ x_4\\ 0\\ 0
\end{bmatrix}
}_{f(x)}
+
\underbrace{
\begin{bmatrix}
0 & 0  \\ 0 & 0  \\ -\tfrac{x_3}{\mathrm{m}} & 0 \\ 0 & -\tfrac{x_4}{\mathrm{m}}
\end{bmatrix}
}_{Y(x)}
\underbrace{
\begin{bmatrix}
\mu_1 \\ \mu_2
\end{bmatrix}
}_{\theta}
+
\underbrace{
\begin{bmatrix}
0 & 0\\ 0 & 0 \\ \tfrac{1}{\mathrm{m}} & 0 \\ 0 & \tfrac{1}{\mathrm{m}}
\end{bmatrix}
}_{g(x)}
\underbrace{
\begin{bmatrix}
u_1 \\ u_2
\end{bmatrix}}_{u},
\end{aligned}
\]
where $\mathrm{m}\in\mathbb{R}_{>0}$ is a known mass and $\mu_1,\mu_2\in\mathbb{R}_{>0}$ are uncertain viscous friction coefficients. For simplicity, we set $\mathrm{m}=\mu_1=\mu_2=1$. The objective is to stabilize the system to the origin while avoiding a set of static obstacles in the state space. To achieve the stabilization task, we construct an ES-aCLF as $V(x)\coloneqq x^\top Px$ with $P=[2,\, 0,\, 1,\, 0; 0,\, 2,\, 0,\, 1; 1,\, 0,\, 1,\, 0; 0,\, 1,\, 0,\, 1]$ and $c_3=1$. The safety objective is achieved by considering a collection of safe sets of the form \eqref{eq:C} with $h^{(i)}(x)\coloneqq (x_1 - p_{1,i})^2 + (x_2 - p_{2,i})^2 - R^2$, where $(p_{1,i},p_{2,i})\in\mathbb{R}^2$ denotes the center of a circular disk and $R_i\in\mathbb{R}_{>0}$ its radius. With straightforward calculations one can verify that this CBF candidate has relative degree $r=2$ with respect to the dynamics and that Assumption \ref{assumption:rd} is satisfied. Although the dynamics and safe set are relatively simple, the  aCBF techniques from \cite{TaylorACC20,LopezLCSS21,DixonACC21,SanfeliceACC21} cannot be applied to solve the problem as currently constructed since the relative degree of $h$ is larger than one\footnote{Note that it may be possible to construct a relative degree one CBF for such a problem by including dependence on velocity into the definition of $h$. As noted earlier, an advantage of HOCBFs is that the design of $\mathcal{C}$ is greatly simplified since dependence on velocity is naturally encoded by higher order terms rather than through the explicit definition of $h$.}. To illustrate the efficacy of the developed approach, we simulate the system under the HO-RaCBF/ES-aCLF control architecture presented herein and compare the results to controllers {\color{black} take a purely robust approach}. The HO-RaCBF/ES-aCLF controller is computed by taking the solution to \eqref{eq:ES-aCLF-QP} as $k_d$ in \eqref{eq:HO-RaCBF-QP}. For each simulation we construct two HOCBFs as outlined previously with $(p_{1,1},p_{2,1})=(-1.75,2)$, $(p_{1,2},p_{2,2})=(-1,0.5)$, and $R_1=R_2=0.5$. For each simulation, the system is initialized as $x(0)=(-2.5,2.5,0,0)$ and all parameter estimates are initialized at zero. All extended class $\mathcal{K}$ functions used are defined as $\alpha(r)=r$. For the learning-based approach, we maintain a history stack with $M=20$ entries, the integration window is chosen as $\Delta T=0.5$, and the learning rates are chosen as $\Gamma=I_2$ and $\gamma=10$.

To demonstrate the relationship between adaptation and safety, we run set of simulations comparing system performance using the HO-RaCBF/ES-aCLF controller to that under a purely robust approach (i.e., accounting for the maximum possible parameter estimation error without adaptation), where the set of possible parameters $\Theta$ is varied, the results of which are presented in Fig. \ref{fig:double_int_traj}-\ref{fig:theta}. The controller used in the robust approach is the same as that of the adaptive approach but with $\gamma=0$. As shown in Fig. \ref{fig:double_int_traj} and Fig. \ref{fig:cbf_traj}, each trajetory is safe; however, for larger levels of uncertainty the purely robust controller is overly conservative, causing the system trajectory to diverge from the origin. In constrast, the adaptive controller reduces the uncertainty online and achieves dual objectives of stability and safety. In fact, the convergence of the trajectory to the origin under the adaptive controller is minimally affected by the initial level of uncertainty, whereas the trajectory under the robust controller fails to converge to the origin in the presence of large uncertainty. The ability of the learning scheme to identify the uncertain parameters is showcased in Fig. \ref{fig:theta}, which shows the trajectory of the estimated parameters used in the ES-aCLF QP and HO-RaCBF QP, both of which converge to their true values in just under 15 seconds.

\begin{figure}
    \centering
    \includegraphics{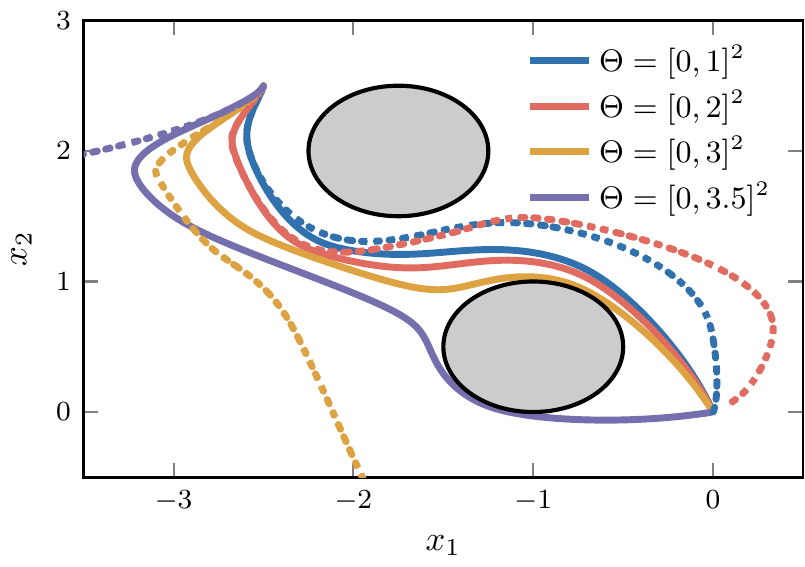}
    \caption{System trajectory under each controller across four different uncertainty sets with solid lines denoting trajectories under the adaptive controller and dashed lines denoting trajectories under the purely robust controller. The gray disks represent the obstacles.}
    \label{fig:double_int_traj}
\end{figure}

\begin{figure}
    \centering
    \includegraphics{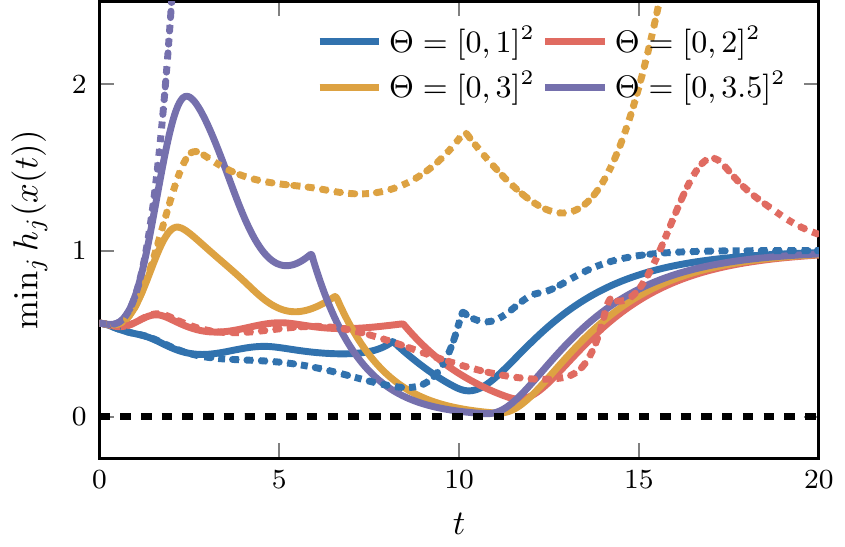}
    \caption{Minimum value among the two HOCBFs point-wise in time along each system trajectory. The solid and dashes curves have the same interpretation as those in Fig. \ref{fig:double_int_traj} and the dashed black line denotes $h(x)=0$.}
    \label{fig:cbf_traj}
\end{figure}

\begin{figure}
    \centering
    \includegraphics{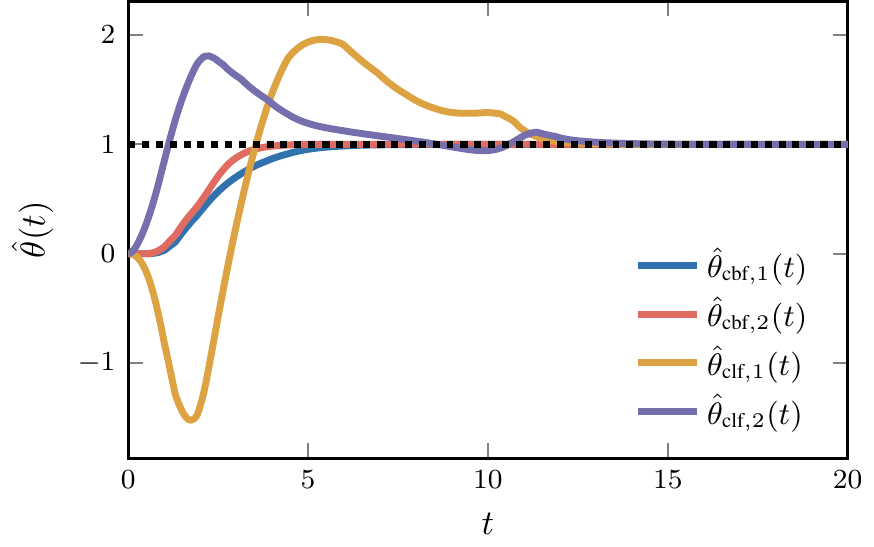}
    \caption{Estimates of the uncertain parameters used in the ES-aCLF QP and HO-RaCBF QP for the simulation corresponding to the uncertainty set $\Theta=[0,3]^2$}
    \label{fig:theta}
\end{figure}

\subsection{Inverted Pendulum}
To demonstrate the applicability of the developed results to an unstable nonlinear system with uncertain parameters, we now consider an inverted pendulum of the form \eqref{eq:dyn2} as
\[
\begin{aligned}
\underbrace{
\begin{bmatrix}
\dot{x}_1 \\ \dot{x}_2
\end{bmatrix}
}_{\dot{x}}
=
\underbrace{
\begin{bmatrix}
x_2 \\ 0
\end{bmatrix}
}_{f(x)}
+
\underbrace{
\begin{bmatrix}
0 & 0  \\  \tfrac{1}{\ell}\sin(x_1) & -\tfrac{1}{\mathrm{m}}x_2
\end{bmatrix}
}_{Y(x)}
\underbrace{
\begin{bmatrix}
\mathrm{g} \\ \mathrm{c}
\end{bmatrix}
}_{\theta}
+
\underbrace{
\begin{bmatrix}
0 \\ \tfrac{1}{\mathrm{m}\ell^2}
\end{bmatrix}
}_{g(x)}
u
\end{aligned},
\]
with length $\ell=0.7$, mass $\mathrm{m}=0.7$, gravitational acceleration $\mathrm{g}=9.8$, and damping coefficient $\mathrm{c}=0.2$. The objective is to regulate $x$ to the origin while satisfying $-\tfrac{\pi}{4}\leq x_1(t)\leq \tfrac{\pi}{4}$ for all time. To achieve the stabilization objective we select the aCLF $V(x)=x^\top Px$ with $P=[1,\,0.5;\,0.5,\,0.5]$ and $c_3=2.5\tfrac{\lambda_{\min}(Q)}{\lambda_{\max}(P)}$, where $Q=[2,\,1;1,\,1]$. The safety objective is achieved by constructing two safe sets defined by $h_1(x)=x_1 + \tfrac{\pi}{4}$ and $h_2(x)=\tfrac{\pi}{4} - x_1$, both of which have relative degree 2 with respect to the dynamics. Again, the results of \cite{TaylorACC20,LopezLCSS21,DixonACC21,SanfeliceACC21} are not applicable to this example since $h$ has relative degree larger than one. The uncertain parameters $\theta=(\theta_1,\theta_2)=(\mathrm{g},\mathrm{c})$ are assumed to take values from $\Theta=[7,13]\times[0,3]\subset\mathbb{R}^2$ and higher order terms from \eqref{eq:psi} are defined using $\alpha_1(r)=\alpha_2(r)=5r$. 

To compare the performance between various control architectures, the system is simulated under the influence of 1) a HO-RaCBF/ES-aCLF controller as in \eqref{eq:HO-RaCBF-QP} with $k_d$ obtained from solving the ES-aCLF QP \eqref{eq:ES-aCLF-QP}; 2) an ES-aCLF controller as in \eqref{eq:ES-aCLF-QP}; 3) an ``open-loop" HO-RaCBF controller as in \eqref{eq:HO-RaCBF-QP} with $k_d\equiv 0$ and an initial condition outside the safe set. For each simulation all parameters associated with the learning scheme remain the same as in the previous example. The results the simulations are presented in Fig. \ref{fig:inv_pen_traj} and Fig. \ref{fig:inv_pen_theta}. Although the ES-aCLF controller achieves the stabilization objective (orange curves in Fig. \ref{fig:inv_pen_traj}), it does so at the cost of violating the safety constraints. In essence, for the ES-aCLF controller to achieve the stabilization objective, it must first allow the pendulum to tip over to collect sufficiently rich data for identifying the dynamics, which, once collected, allows for rapid stabilization to the origin. In contrast, by augmenting the ES-aCLF controller with a HO-RaCBF (blue curves in Fig. \ref{fig:inv_pen_traj}), the stabilization objective is achieved \emph{safely} in an exponential fashion. As predicted by Corollary \ref{corollary:HO-RaCBF}, the HO-RaCBF controller is capable of stabilizing the system to the safe set even when starting from an unsafe initial condition (purple curve in Fig. \ref{fig:inv_pen_traj}). Furthermore, all conservatism associated with the initial parameter uncertainty has essentially been eliminated | once the system is stabilized to the safe set, the HO-RaCBF controller allows the pendulum to lie on the boundary of the safe set without crossing it. Similar to the previous example, the uncertainty in the parameter estimates is exponentially driven to zero (see Fig. \ref{fig:inv_pen_theta}), implying that after only a few seconds of learning, the HO-RaCBF and ES-aCLF conditions from Def. \ref{def:HO-RaCBF} and Def. \ref{def:ES-aCLF}, respectively, closely approximate the original HOCBF condition from Def. \ref{def:HOCBF} and ES-CLF condition from e.g. \cite{AmesTAC14,AmesTAC17,AmesECC19}.

\begin{figure}
    \centering
    \includegraphics{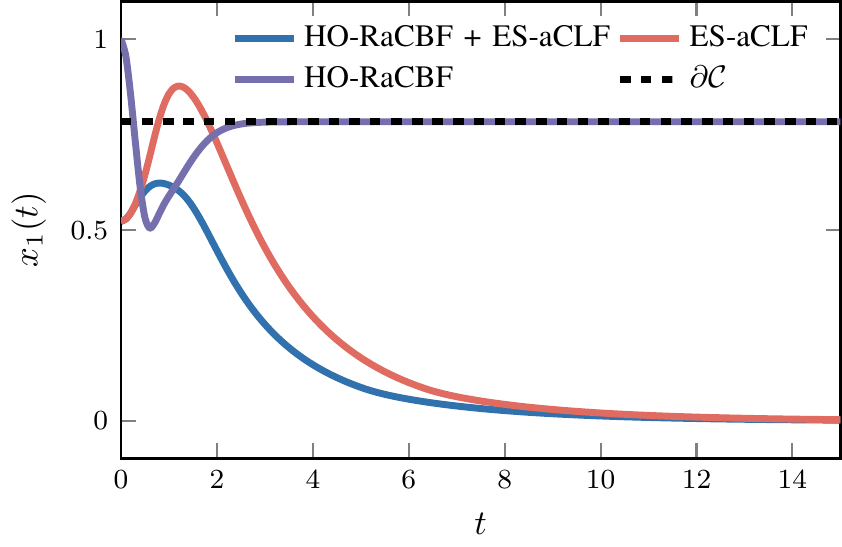}
    \caption{Evolution of the pendulum's orientation under each control architecture. The dashed black line denote the boundary of the safe set.}
    \label{fig:inv_pen_traj}
\end{figure}

\begin{figure}
    \centering
    \includegraphics{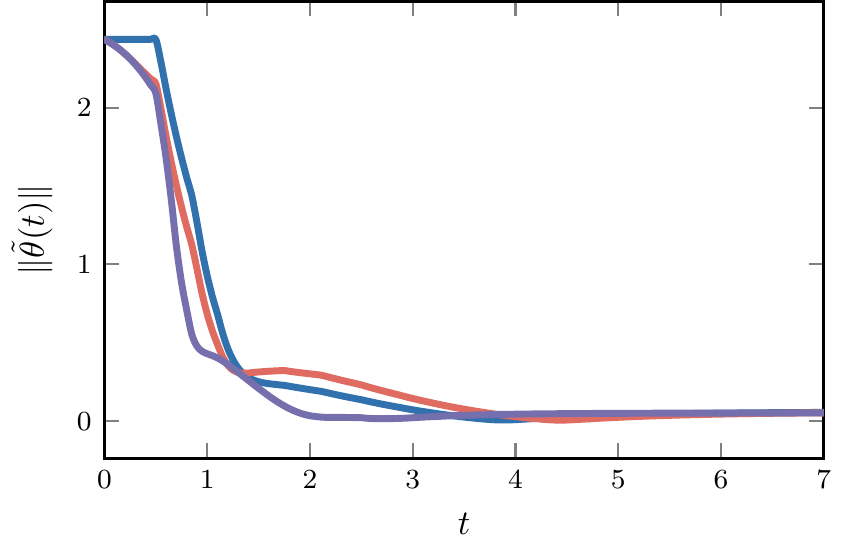}
    \caption{Norm of the parameter estimation error under each control architecture. The color of each curve shares the same interpretation of those in Fig. \ref{fig:inv_pen_traj}.}
    \label{fig:inv_pen_theta}
\end{figure}

\section{Conclusions}
In this paper we introduced HO-RaCBFs and ES-aCLFs as a means to synthesize safe and stable control policies for uncertain nonlinear systems with high relative degree safety constraints. The novel class of HO-RaCBF is, to the best of our knowledge, the first to extend the aCBF paradigm from \cite{TaylorACC20} to CBFs with arbitrary relative degree under mild assumptions regarding the structure of the uncertainty and, unlike existing formulations, the proposed HO-RaCBF inherits desirable robustness properties of zeroing CBFs \cite{AmesADHS15}. The class of ES-aCLFs introduced herein builds upon the classical aCLF formulation \cite[Ch. 4.1]{Krstic} by leveraging data-driven techniques from CL adaptive control to guarantee exponential stability. The advantages of the proposed HO-RaCBFs/ES-aCLFs were illustrated through two numerical examples that cannot be addressed by existing approaches. Directions for future research include relaxing Assumption \ref{assumption:rd} and extending the approach to systems with nonparametric and actuation uncertainty. 

\section*{Acknowledgements}
We thank the anonymous reviewers for their helpful comments and suggestions, which have improved the presentation of the material in paper.

\bibliographystyle{ieeetr}
\bibliography{%
barrier,%
books,%
adaptive,%
hybrid,%
mpc,%
nonlinear%
}

\end{document}